\documentclass[final, 10pt]{article}

\newif\ifpublic
\publicfalse

\newif\ifarxiv
\arxivtrue

\ifpublic
     \usepackage[disable]{todonotes}
\else

	\usepackage[colorinlistoftodos]{todonotes}
\fi

\def\showauthornotes{1}

\def\showdraftbox{1}

\usepackage{ttpalatino}

\usepackage{amsmath,amssymb,amsfonts,bm}
\usepackage{mathtools}
\usepackage{bbold}
\usepackage{epsfig}
\usepackage{latexsym,nicefrac,bbm}
\usepackage{xspace}
\usepackage{color,fancybox,graphicx,url,subfigure}
\usepackage{enumitem}
\usepackage{fullpage}
\usepackage{tabularx}
\usepackage{mdframed}
\usepackage{longtable}
\usepackage{tabu}
\usepackage{stfloats}
\usepackage{framed}
\usepackage{tcolorbox}

\definecolor{mycolor}{rgb}{0.122, 0.435, 0.698}
\makeatletter
\newcommand{\mybox}[1]{%
  \setbox0=\hbox{#1}%
  \setlength{\@tempdima}{\dimexpr\wd0+13pt}%
  \begin{tcolorbox}[colframe=mycolor,boxrule=0.5pt,arc=4pt,
      left=6pt,right=6pt,top=6pt,bottom=6pt,boxsep=0pt,width=\@tempdima]
    #1
  \end{tcolorbox}
}
\makeatother

\usepackage[bookmarks,colorlinks,breaklinks]{hyperref}  
\hypersetup{linkcolor=blue,citecolor=blue,filecolor=blue,urlcolor=blue} 
\renewcommand{\eqref}[1]{\hyperref[#1]{(\ref*{#1})}}

\numberwithin{equation}{section}

\newtheorem{theorem}{Theorem}[section]

\newtheorem{definition}[theorem]{Definition}
\newtheorem{lemma}[theorem]{Lemma}
\newtheorem{remark}[theorem]{Remark}
\newtheorem{proposition}[theorem]{Proposition}

\newtheorem{fact}[theorem]{Fact}

\def\FullBox{\hbox{\vrule width 6pt height 6pt depth 0pt}}

\def\qed{\ifmmode\qquad\FullBox\else{\unskip\nobreak\hfil
\penalty50\hskip1em\null\nobreak\hfil\FullBox
\parfillskip=0pt\finalhyphendemerits=0\endgraf}\fi}

\def\qedsketch{\ifmmode\Box\else{\unskip\nobreak\hfil
\penalty50\hskip1em\null\nobreak\hfil$\Box$
\parfillskip=0pt\finalhyphendemerits=0\endgraf}\fi}

\newenvironment{proof}{\begin{trivlist} \item {\bf Proof:~~}}
   {\qed\end{trivlist}}

\newcommand\R{\mathbb R}

\newcommand{\marginlabel}[1]%
{\mbox{}\marginpar{\it{\raggedleft\hspace{0pt}#1}}}

\definecolor{Mygray}{gray}{0.8}

 \ifcsname ifcommentflag\endcsname\else
  \expandafter\let\csname ifcommentflag\expandafter\endcsname
                  \csname iffalse\endcsname
\fi

\ifnum\showauthornotes=1
\else
\fi

\ifnum\showauthornotes=1
\newcommand{\Authornote}[2]{{\sf\small\color{red}{[#1: #2]}}}
\newcommand{\Authoredit}[2]{{\sf\small\color{red}{[#1]}\color{blue}{#2}}}
\newcommand{\Authorcomment}[2]{{\sf \small\color{Mygray}{[#1: #2]}}}
\newcommand{\Authorfnote}[2]{\footnote{\color{red}{#1: #2}}}
\newcommand{\Authorfixme}[1]{\Authornote{#1}{\textbf{??}}}
\newcommand{\Authormarginmark}[1]{\marginpar{\textcolor{red}{\fbox{
#1:!}}}}
\else
\newcommand{\Authornote}[2]{}
\newcommand{\Authoredit}[2]{}
\newcommand{\Authorcomment}[2]{}
\newcommand{\Authorfnote}[2]{}
\newcommand{\Authorfixme}[1]{}
\newcommand{\Authormarginmark}[1]{}
\fi

\newcommand\calG{\mathcal{G}}
\newcommand\calH{\mathcal{H}}

\newcommand\calD{\mathcal{D}}

\def\implies{\Rightarrow}

\newlength{\pgmtab}  
\newtheorem{Thm}{Theorem}[section]

\newtheorem{observation}[Thm]{Observation}

 {
	\begin{enumerate}}{\end{enumerate}}

\newcounter{lecnum}

\newlength{\tpush}
\ifnum\showdraftbox=1

\else

\fi

\newcommand{\lref}[2][]{\hyperref[#2]{#1~\ref*{#2}}}
\renewcommand{\eqref}[1]{\hyperref[#1]{(\ref*{#1})}}

\numberwithin{equation}{section}

\newcommand{\Var}{{\sf Var}}
 
\newcommand{\err}{\mathsf{err}}
\newcommand{\Inf}{\mathsf{Inf}}

\newcommand{\term}{\mathsf{Term}}

\newcommand{\calX}{\mathcal{X}}

\newcommand{\calT}{\mathcal{T}}

\newcommand{\calN}{\mathcal{N}}
\DeclareMathOperator*{\E}{\mathbf{E}}

\renewcommand{\R}{\mathbf{R}}
\renewcommand{\P}{\mathcal{P}}
\newcommand{\F}{\mathbf{F}}
\renewcommand{\E}{\mathop{\mathbf{E}}}
\newcommand{\cald}{\mathcal{D}}
\renewcommand{\inf}{\mathtt{inf}}

\newcommand{\X}{\mathcal{X}}

\renewcommand{\v}[1]{\bm{#1}}

\title{An Improved Dictatorship Test with Perfect Completeness}

\author{Amey Bhangale
\thanks{Department of Computer Science, Rutgers University, USA }
\and 
Subhash Khot\thanks{Computer Science Department, New York University, New York, USA}
\and 
Devanathan Thiruvenkatachari\thanks{Computer Science Department, New York University, New York, USA}
}

\emergencystretch=\maxdimen
\begin{document}

\maketitle
\newcommand{\gapLC}{\mathtt{GapLC}}
\newcommand{\gapUG}{\mathtt{GapUG}}

\begin{abstract}
A Boolean function $f:\{0,1\}^n\rightarrow \{0,1\}$ is called a dictator if it depends on exactly one variable i.e $f(x_1, x_2, \ldots, x_n) = x_i$ for some $i\in [n]$. In this work, we study a $k$-query dictatorship test. Dictatorship tests are central in proving many hardness results for constraint satisfaction problems. 

The dictatorship test is said to have {\em perfect completeness} if it accepts any dictator function. The {\em soundness} of a test is the maximum probability with which it accepts any function far from a dictator. Our main result is a $k$-query dictatorship test with perfect completeness and soundness $ \frac{2k + 1}{2^k}$, where $k$ is of the form $2^t -1$ for any integer $t > 2$. This improves upon the result of \cite{TY15} which gave a dictatorship test with soundness $ \frac{2k + 3}{2^k}$.
\end{abstract}

\section{Introduction}
Boolean functions are the most basic objects in the field of theoretical computer science. Studying different properties of Boolean functions has found applications in many areas including hardness of approximation, communication complexity, circuit complexity etc. In this paper, we are interested in studying Boolean functions from a property testing point of view. 

In {\em property testing}, one has given access to a function $f :\{0, 1\}^n \rightarrow \{0, 1\}$ and the task is to decide if a given function has a particular property or whether it is {\em far} from it. One natural notion of farness is what fraction of $f$'s output we need to change so that the modified function has the required property. A verifier can have an access to random bits. This task of property testing seems trivial if we do not have restrictions on how many queries one can make and also on the computation. One of the main questions in this area is can we still decide if $f$ is very far from having the property by looking at a very few locations with high probability. 

There are few different parameters which are of interests while designing such tests including the amount of randomness, the number of locations queried, the amount of computation the verifier is allowed to do etc. The test can either be {\em adaptive} or {\em non-adaptive}. In an adaptive test, the verifier is allowed to query a function at a few locations and based on the answers that it gets, the verifier can decide the next locations to query whereas a non-adaptive verifier queries the function in one shot and once the answers are received makes a decision whether the function has the given property. In terms of how good the prediction is we want the test to satisfy the following two properties:
\begin{itemize}
\item {\bf Completeness:} If a given function has the property then the test should accept with high probability
\item {\bf Soundness:} If the function is far from the property then the test should accept with very tiny probability.
\end{itemize}
A test is said to have {\em perfect completeness} if in the completeness case the test always accepts. A test with {\em imperfect completeness} (or almost perfect completeness) accepts a dictator function with probability arbitrarily close to $1$. Let us define the soundness parameter of the test as how small we can make the acceptance probability in the soundness case. \\

A function is called a {\em dictator} if it depends on exactly one variable i.e $f(x_1, x_2, \ldots, x_n) = x_i$ for some $i\in [n]$. In this work, we are interested in a non-adaptive test with perfect completeness which decides whether a given function is a dictator or far from it.  This was first studied in \cite{BellareGS98, Parnas02} under the name of Dictatorship test and Long Code test. Apart from a natural property, dictatorship test has been used extensively in the construction of probabilistically checkable proofs (PCPs) and hardness of approximation. 

An instance of a {\em Label Cover} is a bipartite graph $G((A,B),E)$ where each edge $e\in E$ is labeled by a projection constraint $\pi_e : [L]\rightarrow [R]$. The goal is to assign labels from $[L]$ and $[R]$ to vertices in $A$ and $B$ respectivels so that the number of edge constraints satisfied is maximized. Let $\gapLC(1, \epsilon)$ is a promise gap problem where the task is to distinguish between the case when all the edges can be satisfied and at most $\epsilon$ fraction of edges are satisfied by any assignment. As a consequence of the PCP Theorem \cite{AroraLMSS1998, AroraS1998} and the Parallel Repetition Theorem\cite{Raz1998}, $\gapLC(1, \epsilon)$ is NP-hard for any constant $\epsilon>0$.   In \cite{Hastad2001}, H{\aa}stad used various dictatorship tests along with the hardness of Label Cover to prove optimal inapproximability results for many constraint satisfaction problems. Since then dictatorship test has been central in proving hardness of approximation.

A dictatorship test with $k$ queries and $P$ as an accepting predicate is usually useful in  showing hardness of approximating Max-$P$ problem.  Although this is true for many CSPs, there is no black-box reduction from such dictatorship test to getting inapproximability result. One of the main obstacles in converting dictatorship test to NP-hardness result is that the constraints in Label Cover are $d$-to-$1$ where the the  parameter $d$ depends on $\epsilon$ in $\gapLC(1, \epsilon)$. To remedy this, Khot in \cite{Khot02UGC} conjectured that a Label Cover where the constraints are $1$-to-$1$, called {\em Unique Games}, is also hard to approximate within any constant.  More specifically, Khot conjectured that $\gapUG(1-\epsilon, \epsilon)$, an analogous promise problem for Unique Games,  is NP-hard for any constant $\epsilon>0$.  One of the significance of this conjecture is that many dictatorship tests can be composed easily with  $\gapUG(1-\epsilon, \epsilon)$ to get  inapproximability results. However, since the Unique Games problem lacks perfect completeness it cannot be used to show hardness of approximating {\em satisfying} instances.

From the PCP point of view, in order to get $k$-bit PCP with perfect completeness, the first step is to analyze $k$-query dictatorship test with perfect completeness. For its application to construction PCPs there are two important things we need to study about the dictatorship test. First one is how to compose the dictatorship test with the known PCPs and second is how sound we can make the dictatorship test. In this work, we make a progress in understanding the answer to the later question. To make a remark on the first question, there is a dictatorship test with perfect completeness and soundness $\frac{2^{\tilde{O}(k^{1/3})}}{2^k}$ and also a way to compose it with $\gapLC(1, \epsilon)$ to get a $k$-bit PCP with perfect completeness and the same soundness that of the dictatorship test. This was done in \cite{Huang13} and is currently the best know $k$-bit non-adaptive PCP with perfect completeness.

\paragraph{Distance from a dictator function:} There are multiple notion of closeness to a dictator function. One natural definition is the minimum fraction of values we need to change such that the function becomes a dictator. There are other relaxed notions such as how close the function is to {\em juntas} - functions that depend on constantly many variables. Since our main motivation is the use of dictatorship test in the construction of PCP, we can work with even more relaxed notion which we describe next: For a Boolean function $f:\{0,1\}^n \rightarrow \{0,1\}$ an influence of $i^{th}$ variable is the probability that for a random input $x\in \{0,1\}^n$ flipping the $i^{th}$ coordinate flips the value of the function. Note that a dictator function has a variable whose influence is $1$. The influence of $i^{th}$ variable can be expressed in terms of the fourier coefficients of $f$ as $\inf_i[f] = \sum_{S\subseteq [n] \mid i\in S}\hat{f}(S)^2$. Using this, a degree $d$ influence of $f$ is $\inf_i^{\leq d}[f] = \sum_{S \subseteq [n] \mid i\in S, |S|\leq d}\hat{f}(S)^2$. We say that $f$ is far from any dictator if for a constant $d$ all its degree $d$ influences are upper bounded by some small constant.\\

In this paper, we investigate the trade-off between the number of queries and the soundness parameter of a dictatorship test with perfect completeness w.r.t to the above defined distance to a dictator function.  A random function is far from any dictator but still it passes any (non-trivial) $k$-query test with probability at least $1/2^k$. Thus, we cannot expect the test to have soundness parameter less than $1/2^k$. The main theorem in this paper is to show there exists a dictatorship test with perfect completeness and soundness at most $\frac{2k+1}{2^k}$.

\begin{theorem}
Given a Boolean function $f : \{0, 1\}^n \rightarrow \{0, 1\}$, for every $k$ of the form $2^m-1$ for any $m>2$, there is a $k$ query dictatorship test with perfect completeness and soundness $\frac{2k+1}{2^k}$.
\end{theorem}

Our theorem improves a result of Tamaki-Yoshida\cite{TY15} which had a soundness of $\frac{2k+3}{2^k}$.

\begin{remark}
Tamaki-Yoshida~\cite{TY15} studied a $k$ functions test where if a given set of $k$ functions are all the same dictator then the test accepts with probability $1$. They use low degree {\em cross influence} (Definition 2.4 in ~\cite{TY15}) as a criteria to decide closeness to a dictator function. Our whole analysis also goes through under the same setting as that of~\cite{TY15}, but we stick to single function version for a cleaner presentation.
\end{remark}

\subsection{Previous Work}
The notion of Dictatorship Test was introduced by Bellare et al.~\cite{BellareGS98} in the context of Probabilistically Checkable Proofs and also studied by Parnas et al. ~\cite{Parnas02}. As our focus is on non-adpative test, for an adaptive $k$-bit dictatorship test, we refer interested readers to ~\cite{ST09, HW03, HK05, EH08}. Throughout this section, we use $k$ to denote the number of queries and $\epsilon>0$ an arbitrary small constant.

Getting the soundness parameter for a specific values of $k$ had been studied earlier. For instance, for $k=3$ H{\aa}stad~\cite{Hastad2001} gave a $3$-bit PCP with completeness $1-\epsilon$ and soundness $1/2+\epsilon$. It was earlier shown by Zwick~\cite{Zwick97} that any $3$-bit dictator test with perfect completeness must have soundness at at least $5/8$. For a $3$-bit dictatorship test with perfect completeness,  Khot-Saket~\cite{KS06} acheived a soundness parameter $20/27$ and they were also able to compose their test with Label Cover towards getting $3$-bit PCP with similar completeness and soundness parameters. The dictatorship test of Khot-Saket~\cite{KS06} was later improved by O'Donnell-Wu~\cite{OW09} to the optimal value of $5/8$. The dictatorship test of O'Donnell-Wu~\cite{OW09} was used in O'Donnell-Wu~\cite{OW09conditional} to get a conditional (based on Khot's $d$-to-$1$ conjecture) $3$-bit PCP with perfect completeness and soundness $5/8$ which was later made unconditional by H{\aa}stad~\cite{H14np}. 

For a general $k$, Samorodensky-Trevisan~\cite{ST00PCP} constructed a $k$-bit PCP with imperfect completeness and soundness $2^{2\sqrt{k}}/2^k$.  This was improved later by Engebretsen and Holmerin~\cite{EH08} to $2^{\sqrt{2k}}/2^k$ and by H{\aa}stad-Khot~\cite{HK05} to $2^{4\sqrt{k}}/2^k$ with perfect completeness. To break the  $2^{O(\sqrt{k})}/2^k$ Samorodensky-Trevisan~\cite{ST09} introduced the relaxed notion of soundness (based on the low degree influences) and gave a dictatorship test (called Hypergraph dictatorship test) with almost perfect completeness and soundness $2k/2^k$ for every $k$ and also $(k+1)/2^k$ for infinitely many $k$. They combined this test with Khot's Unique Games Conjecture~\cite{Khot02UGC} to get a conditional $k$-bit PCP with similar completeness and soundness guarantees. This result was improved by Austrin-Mossel~\cite{AustrinM2009} and they achieved $k+o(k)/2^k$ soundness.

For any $k$-bit CSP for which there is an instance with an integrality gap of $c/s$ for a certain SDP, using a result of Raghavendra~\cite{R08} one can get a dictatorship test with completeness $c-\epsilon$ and soundness $s+\epsilon$. Getting the explicit values of $c$ and $s$ for a given value of $k$ is not clear from this result and also it cannot be used to get a dictatorship test with perfect completeness. Similarly, using the characterization of strong approximation restance of Khot et. al~\cite{KTW14} one can get a dictatorship test but it also lacks peferct completeness. Recently, Chan~\cite{Chan13} significantly improved the parameters for a $k$-bit PCP which achieves soundness $2k/2^k$ albeit losing perfect completeness. Later Huang~\cite{Huang13} gave a $k$-bit PCP with perfect completeness and soundness $2^{\tilde{O}(k^{1/3})}/{2^k}$. 

As noted earlier, the previously best known result for a $k$-bit dictatorship test with perfect completeness is by Tamaki-Yoshida~\cite{TY15}. They gave a test with soundness $\frac{2k+3}{2^k}$ for infinitely many $k$.
\subsection{Proof Overview}

Let $f :\{-1,+1\}^n \rightarrow \{-1,+1\}$ be a given balanced Boolean function \footnote{Here we switch from $0/1$ to $+1/-1$ for convenience. With this notation switch, balanced function means $\E[f(\v x)] = 0$}. Any  non-adaptive $k$-query dictatorship test queries the function $f$ at $k$ locations and receives $k$ bits which are the function output on these queries inputs. The verifier then applies some predicate, let's call it $\P :\{0,1\}^k \rightarrow \{0,1\}$, to the received bits and based on the outcome decides whether the function is a dictator or far from it. Since we are interested in a test with perfect completeness this puts some restriction on the set of $k$ queried locations. If we denote $\v x_1, \v x_2, \ldots, \v x_k$ as the set of queried locations then the $i^{th}$ bit from $(\v x_1, \v x_2, \ldots, \v x_k)$ should satisfy the predicate $\P$. This is because, the test should always accept no matter which dictator $f$ is. 

Let $\mu$ denotes a distribution on $\P^{-1}(1)$. One natural way to sample $(\v x_1, \v x_2, \ldots, \v x_k)$ such that the test has a perfect completeness guarantee is for each coordinate $i\in [n]$ independently sample $(\v x_1, \v x_2, \ldots, \v x_k)_i$ from distribution $\mu$. This is what we do in our dictatorship test for a specific distribution $\mu$ supported on $\P^{-1}(1)$. It is now easy to see that the test accepts with probability $1$ of $f$ is an $i^{th}$ dictator for any $i\in [n]$.

Analyzing the soundness of a test is the main technical task. First note that the soundness parameter of the test depends on $\P^{-1}(1)$ as it can be easily verified that if $f$ is a random function, which is far from any dictator function, then the test accepts with probability at least $\frac{|\P^{-1}(1)|}{2^k}$. Thus, for a better soundness guarantee we want $P$ to have as small support as possible. The acceptance probability of the test is given by the following expression:
\begin{align*}
    \Pr[\mbox{Test accepts } f] &= \E [\P(f(\v{x}_1), f(\v{x}_2),\cdots, f(\v{x}_k))]\\
    & = \frac{|\P^{-1}(1)|}{2^k} + \E\left[\sum_{S\subseteq [k], S\neq \emptyset}\hat{\P}(S) \prod_{i\in S}f(\v{x}_i)\right]
\end{align*}
Thus, in order to show that the test accepts with probability at most $\frac{|\P^{-1}(1)|}{2^k} + \epsilon$ it is enough to show that all the expectations $E_S:=|\E[\prod_{i\in S}f(\v{x}_i) ]|$ are small if $f$ is far from any dictator function. Recall that at this point, we can have any predicate $\P$ on $k$ bits which the verifier uses. As we will see later, for the soundness analysis we need the predicate $\P$ to satisfy certain properties. 

For the rest of the section, assume that the given function $f$ is such that the low degree influence of every variable $i\in [n]$ is very small constant $\tau$. If $f$ is a constant degree function (independent of $n$) then the usual analysis goes by invoking invariance principle to claim that the quantity $E_S$ does not change by much if we replace the distribution $\mu$ to a distribution $\xi$ over Gaussian random variable with the same first and second moments.  An advantage of moving to a Gaussian distribution is that if $\mu$ was a uniform and pairwise independent distribution then so is $\xi$ and using the fact that a pairwise independence implies a total independence in the Gaussian setting, we have $E_S \approx |\prod_{i\in S}\E[f(\v{g}_i)]|$. Since we assumed that $f$ was a balanced function we have $\E[f(\v{g}_i)]| = 0$ and hence we can say that the quantity $E_S$ is very small. 

There are two main things we need to take care in the above argument. $1)$ We assumed that $f$ is a low degree function and in general it may not be true. $2)$ The argument crucially needed $\mu$ to satisfy pairwise independence condition and hence it puts some restriction on the size of $\P^{-1}(1)$ (Ideally, we would like $|\P^{-1}(1)|$ to be as small as possible for a better soundness guarantee). We take care of $(1)$, as in the previous works \cite{TY15, OW09, AustrinM2009} etc., by requiring the distribution $\mu$ to have {\em correlation} bounded away from $1$. This can be achieved by making sure the support of $\mu$ is {\em connected} - for every coordinate $i\in[k]$ there exists $a, b\in \P^{-1}(1)$ which differ at the $i^{th}$ location. For such distribution, we can add independent {\em noise} to each co-ordinate without changing the quantity $E_S$ by much. Adding independent noise has the effect that it damps the higher order fourier coefficients of $f$ and the function behaves as a low degree function. We can now apply invariance principle to claim that $E_S\approx 0$. This was the approach in \cite{TY15} and they could find a distribution $\mu$ whose support size is $2k+3$ which is connected and pairwise independent. 

In order to get an improvement in the soundness guarantee, our main technical contribution is that we can still get the overall soundness analysis to go through even if $\mu$ does not support pairwise independence condition. To this end, we start with a distribution $\mu$ whose support size is $2k+1$ and has the property that it is {\em almost} pairwise independent. Since we lack pairwise independence, it introduces few obstacles in the above mentioned analysis. First, the {\em amount} of noise we can add to each co-ordinate has some limitations. Second, because of the limited amount of independent noise, we can no longer say that the function $f$ behaves as a low degree function after adding the noise. With the limited amount of noise, we can say that $f$ behaves as a low degree function as long as it does not have a large fourier mass in some interval i.e the fourier mass corresponding to $\hat{f}(T)^2$ such that $|T|\in (s, S)$ for some constant sized interval $(s,S)$ independent of $n$. We handle this obstacle by designing a family of distributions $\mu_1, \mu_2, \ldots, \mu_r$ for large enough $r$ such that the intervals that we cannot handle for different $\mu_i$'s are disjoint. Also, each $\mu_i$ has the same support and is almost pairwise independent. We then let our final test distribution as first selecting $i\in [r]$ u.a.r and then doing the test with the corresponding distribution $\mu_i$. Since the total fourier mass of a $-1/+1$ function is bounded by $1$ and $f$ was fixed before running the test it is very unlikely that $f$ has a large fourier mass in the interval corresponding to the  selected distribution $\mu_i$. Hence, we can conclude that for this overall distribution, $f$ behaves as a low degree function. We note that this approach of using family of distributions was used in \cite{H14np} to construct a $3$-bit PCP with perfect completeness. There it was used in the composition step.

To finish the soundness analysis, let $\tilde{f}$ be the low degree part of $f$. The argument in the previous paragraph concludes that $E_S \approx  |\E[\prod_{i\in S}\tilde{f}(\v{x}_i)]|$. As in the previous work, we can now apply invariance principle to claim that $E_S \approx |\E[\prod_{i\in S}\tilde{f}(\v{g}_i)]|$ where the $i^{th}$ coordinate $(\v g_1, \v g_2, \ldots, \v g_k)_i$ is distributed according to $\xi$ which is almost pairwise independent. We can no longer bring the expectation inside as our distribution lacks independence. To our rescue, we have that the degree of $\tilde{f}$ is bounded by some constant independent of $n$. We then prove that low degree functions are robust w.r.t slight perturbation in the inputs on average. This lets us conclude $\E[\prod_{i\in S}\tilde{f}(\v{g}_i)] \approx \E[\prod_{i\in S}\tilde{f}(\v{h}_i)]$ where $(\v h_1, \v h_2, \ldots, \v h_k)_i$ is pairwise independent. We now  use the property of independence of Gaussian distribution and bring the expectation inside to conclude that $E_S \approx  |\E[\prod_{i\in S}\tilde{f}(\v{h}_i)] |= |\prod_{i\in S}\E[\tilde{f}(\v{h}_i)]| = 0$.

\section{Organization}
We start with some preliminaries in \lref[Section]{section:prelims}. In \lref[Section]{section:test} we describe our dictatorship test. Finally, in \lref[Section]{section:analysis} we prove the analysis of the described dictatorship test.
\section{Preliminaries}
\label{section:prelims}
For a positive integer $k$, we will denote the set $\{1, 2, \ldots, k\}$ by $[k]$. For a distribution $\mu$, let $\mu^{\otimes n}$ denotes the $n$-wise product distribution.


\subsection{Analysis of Boolean Function over Probability Spaces} For
a function $f:\{0,1\}^n \rightarrow \R$, the \emph{Fourier
decomposition} of $f$ is given by 
$$f(x) = \sum_{T\subseteq [n]}
\widehat f(T) \chi_T(x) \text{ where } \chi_T(x) :=\prod_{i\in T} (-1)^{x_i} \text{ and }\widehat f(T) :=
\E_{x \in \{0,1\}^n} f(x)\chi_T(x).$$
The \emph{Efron-Stein decomposition} is
a generalization of the Fourier decomposition to product distributions
of arbitrary probability spaces.

\begin{definition}
Let $(\Omega, \mu)$ be a probability space and
$(\Omega^n,\mu^{\otimes n})$ be the corresponding product space. For a function
$f:\Omega^n\rightarrow \R$, the Efron-Stein decomposition of $f$ with
respect to the product space is given by
$$ f(x_1,\cdots, x_n) = \sum_{\beta \subseteq [n]} f_\beta(x),$$
where $f_\beta$ depends only on $x_i$ for $i\in \beta$ and for all
$\beta' \not\supseteq \beta , a \in \Omega^{\beta'}$, $\E_{x \in
\mu^{\otimes n}} \left[ f_\beta(x) \mid x_{\beta'} = a \right]=0$. 
\end{definition}

Let $\|f\|_p :=
\E_{x\in \mu^{\otimes n}}[|f(x)|^p]^{1/p}$ for $1\leq p<\infty$ and $\|f\|_\infty :=\max_{x\in \Omega^{\otimes n}}|f(x)|$ . 

\begin{definition}
For a multilinear polynomial $f : \R^n \rightarrow \R$ and any $D\in [n]$ define
$$ f^{\leq D} := \sum_{T\subseteq [n], |T|\leq D} \hat{f}(T) \chi_T$$
i.e. $f^{\leq D}$ is degree $D$ part of $f$. Also define $f^{>D} := f-f^{\leq D}$.
\end{definition}

\begin{definition}
For $i \in [n]$, the influence of the $i$th
coordinate on $f$ is defined as follows.
$$\Inf_i[f] := \E_{x_1,\cdots, x_{i-1},x_{i+1},\cdots , x_n}\Var_{x_i}[f
(x_1,\cdots, x_n)]  = \sum_{\beta: i\in \beta} \|f_\beta\|^2_2.$$
For an integer $d$, the degree $d$ influence is defined as
$$\Inf_i^{\leq d}[f] := \sum_{\beta: i\in \beta, |\beta| \leq d} \|f_\beta\|^2_2.$$
\end{definition}
It is easy to see that for Boolean functions, the sum of all the degree $d$ influences is at most $d$.
A dictator is a function which depends on one variable. Thus, the degree $1$ influence of any dictator function is $1$ for some $i\in [n]$. We call a function {\em far} from any dictator if for every $i\in [n]$, the degree $d$ influence is very small for some large $d$. This motivates the following definition.
\begin{definition}[$(d, \tau)$-quasirandom function]
A multilinear function $f:\R^n \rightarrow \R$ is said to be $(d,\tau)$-quasirandom if for every $i\in[n]$ it holds that
$$ \sum_{i\in S\subseteq[n] ,|S|\leq d} \hat{f}(S)^2 \leq \tau$$
\end{definition}

We recall the Bonami-Beckner operator on Boolean functions.

\begin{definition}
For $\gamma \in [0,1]$, the Bonami-Beckner operator $T_{1-\gamma}$ is a linear operator mapping functions $f:\{0,1\}^n \rightarrow \R$ to functions $T_{1-\gamma}f : \{0, 1\}^n \rightarrow \R$ as $T_{1-\gamma}f(x) = \E_{y}[f(y)]$ where $y$ is sampled by setting $y_i = x_i$ with probability $1-\gamma$ and $y_i$ to be uniformly random bit with probability $\gamma$ for each $i\in [n]$ independently. 
\end{definition}

We have the following relation between the fourier decomposition of $T_{1-\gamma}f$ and $f$.
\begin{fact}
\label{fact:noise fd}
$T_{1-\gamma}f = \sum_{T\subseteq [n]} (1-\gamma)^{|T|}\hat{f}(T)\chi_T$.
\end{fact}

\subsection{Correlated Spaces}

Let $\Omega_1\times \Omega_2$ be two correlated spaces and $\mu$ denotes the joint distribution. Let $\mu_1$ and $\mu_2$ denote the marginal of $\mu$ on space $\Omega_1$ and $\Omega_2$ respectively. The correlated space $\rho(\Omega_1\times \Omega_2; \mu)$ can be represented as a bipartite graph on $(\Omega_1, \Omega_2)$ where $x\in \Omega_1$ is connected to $y\in \Omega_2$ iff $\mu(x,y) > 0$. We say that the correlated spaces is $connected$ if this underlying graph is connected.

We need a few definitions and
lemmas related to correlated spaces defined by Mossel~\cite{Mossel2010}.
\begin{definition}
\label{def:correlation}
Let $(\Omega_1 \times \Omega_2, \mu)$ be a finite correlated space, the correlation between $\Omega_1$ and $\Omega_2$ with respect to $\mu$ us defined as 
$$\rho(\Omega_1, \Omega_2; \mu) := \mathop{\max}_{\substack{f : \Omega_1 \rightarrow \R, \E[f]  = 0 , \E[f^2]\leq 1 \\ g: \Omega_2 \rightarrow \R, \E[g] = 0 , \E[g^2]\leq 1} }  \E_{(x,y) \sim \mu }[ |f(x)g(y)|] .$$
\end{definition}

\noindent The following result (from \cite{Mossel2010}) provides a way to upper bound correlation of a correlated spaces.

\begin{lemma}
\label{lemma:corr_bound}
Let $(\Omega_1\times \Omega_2, \mu)$ be a finite correlated space such that the probability of the smallest atom in $\Omega_1\times \Omega_2$ is at least $\alpha>0$ and the correlated space is connected then 
$$ \rho(\Omega_1, \Omega_2; \mu) \leq 1-\alpha^2/2 $$
\end{lemma}

\begin{definition}[Markov Operator]
\label{def:markovop}
Let $(\Omega_1 \times \Omega_2, \mu)$ be a finite correlated space, the Markov operator, associated with this space, denoted by $U$, maps a function $g : \Omega_2 \rightarrow \R$ to functions $Ug : \Omega_1 \rightarrow \R$ by the following map:
$$ (Ug) (x) := \E_{(X,Y)\sim \mu}[g(Y) \mid X=x ].$$
\end{definition}

\noindent In the soundness analysis of our dictatorship test, we will need to understand the Efron-Stein decomposition of $Ug$ in terms of the decomposition of $g$. The following proposition gives a way to relate these two decompositions.
\begin{proposition}[{\cite[Proposition~2.11]{Mossel2010}}]
\label{prop:mossel_prop211}
Let  $(\prod_{i=1}^n\Omega_i^{(1)} \times \prod_{i=1}^n\Omega_i^{(2)}, \prod_{i=1}^n\mu_i)$ be a product correlated spaces. Let $g :\prod_{i=1}^n\Omega_i^{(2)} \rightarrow \mathbf{R}$ be a function and $U$ be the Markov operator mapping functions form space $\prod_{i=1}^n\Omega_i^{(2)}$ to the functions on space $\prod_{i=1}^n\Omega_i^{(1)}$. If $g = \sum_{S\subseteq [n]}g_S$ and $Ug = \sum_{S\subseteq [n]} (Ug)_S$ be the Efron-Stein decomposition of $g$ and $Ug$ respectively then,
$$ (Ug)_S = U(g_S)$$
i.e. the Efron-Stein decomposition commutes with Markov operators.
\end{proposition}

\noindent Finally, the following proposition says that if the correlation between two spaces is bounded away from $1$ then {\em higher order} terms in the Efron-Stein decomposition of $Ug$ has a very small $\ell_2$ norm compared to the $\ell_2$ norm of the corresponding higher order terms in the Efron-Stein decomposition of $g$. 
\begin{proposition}[{\cite[Proposition~2.12]{Mossel2010}}]
\label{prop:mossel_prop212}
Assume the setting of \lref[Proposition]{prop:mossel_prop211} and
furthermore assume that $\rho(\Omega_i^{(1)}, \Omega_i^{(2)}; \mu_i)
\leq \rho$ for all $i\in [n]$, then for all $g$ it holds that
$$ \|U(g_S) \|_2 \leq \rho^{|S|}\|g_S\|_2.$$
\end{proposition}

\subsection{Hypercontractivity}
\begin{definition}
A random variable $r$ is said to be $(p,q,\eta)$-hypercontractive if it satisfies 
$$ \| a + \eta r\|_q \leq \| a + r \|_p$$
for all $a\in \R$.
\end{definition}

We note down the hypercontractive parameters for Rademacher random variable (uniform over $\pm 1$) and standard gaussian random variable. 
\begin{theorem}[\cite{Wol07}\cite{O03}]
\label{thm:hc param}
Let $X$ denote either a uniformly random $\pm 1$ bit, a standard one-dimensional Gaussian. Then $X$ is $\left(2,q,\frac{1}{\sqrt{q-1}}\right)$-hypercontractive.
\end{theorem}

The following proposition says that the higher norm of a low degree function  w.r.t hypercontractive sequence of ensembles is bounded above by its second norm.
\begin{proposition}[\cite{MOO05}]
\label{prop:hyper}
Let $\v x$ be a $(2,q,\eta)$-hypercontractive sequence of ensembles and $Q$ be a multilinear polynomial of degree $d$. Then
$$\|Q(\v x)\|_q \leq \eta^{-d} \|Q(\v x)\|_2 $$
\end{proposition}

\subsection{Invariance Principle}

Let $\mu$ be any distribution on $\{-1,+1\}^k$. Consider the following distribution on $\v{x}_1, \v{x}_2,\ldots, \v{x}_k \in \{-1,+1\}^n$ such that independently for each $i\in[n]$, $((\v{x}_1)_i, (\v{x}_2)_i,\ldots, (\v{x}_k)_i)$ is sampled from $\mu$. We will denote this distribution as $\mu^{\otimes n}$. We are interested in evaluation of a multilinear polynomial $f:\R^n \rightarrow\R$ on $(\v{x}_1, \v{x}_2,\ldots, \v{x}_k)$ sampled as above. 

\noindent Invariance principle shows the closeness between two different distributions w.r.t some quantity of interest. We are now ready to state the version of the invariance principle from \cite{Mossel2010} that we need.
\begin{theorem}[\cite{Mossel2010}]
\label{thm:invariance}
For any $\alpha>0, \epsilon>0, k\in \mathbf{N}^+$ there are $d,\tau>0$ such that the following holds:
Let $\mu$ be the distribution on $\{+1,-1\}^k$ satisfying
\begin{enumerate}
\item $\E_{x\sim \mu}[x_i] = 0$ for every $i\in [k]$
\item $\mu(x)\geq \alpha $ for every $x\in\{-1,+1\}^k$ such that $\mu(x)\neq 0$
\end{enumerate}
Let $\nu$ be a distribution on standard jointly distributed Gaussian variables with the same covariance matrix as distribution $\mu$. Then, for every set of $k$ $(d,\tau)$-quasirandom multilinear polynomials $f_i:\R^n \rightarrow \R$, and suppose $\Var[f_i^{>d}]\leq (1-\gamma)^{2d}$ for $0 < \gamma<1$ it holds that 
$$\left|\E_{(\v{x}_1, \v{x}_2,\ldots, \v{x}_k)\sim \mu^{\otimes n}} \left[ \prod_{i=1}^k f_i(\v{x}_i)\right] - \E_{(\v{g}_1, \v{g}_2,\ldots, \v{g}_k)\sim \nu^{\otimes n}} \left[ \prod_{i=1}^k f_i(\v{g}_i)\right]\right| \leq \epsilon $$
(Note: one can take $d = \frac{\log(1/\tau)}{\log(1/\alpha)}$ and $\tau$ such that $\epsilon = \tau^{\Omega(\gamma/\log(1/\alpha))}$, where $\Omega(.)$ hides constant depending only on $k$.)
\end{theorem}


\section{Query efficient Dictatorship Test}
\label{section:test}
We are now ready to describe our dictatorship test. The test queries a function at $k$ locations and based on the $k$ bits received decides if the function is a dictator or far from it. The check on the received $k$ bits is based on a predicate with few accepting inputs which we describe next.

\subsection{The Predicate}
Let $k = 2^m-1$ for some $m > 2$. Let the coordinates of the predicate is indexed by elements of $\F_2^m\setminus \v{0}=: \{w_1, w_2, \ldots, w_{2^m-1}\}$. The Hadamard predicate $H_k$ has following satisfying assignments:

$$H_k = \{ x\in \{0,1\}^k | \exists a\in \F_2^m\setminus {\v 0} \mbox{ s.t } \forall i\in[k], x_i = a\cdot w_i \}$$
We will identify the set of satisfying assignments in $H_k$ with the variables $h_1, h_2, \ldots,h_k$.

Our final predicate $\P_k$ is the above predicate along with few more satisfying assignments. More precisely, we add all the assignments which are at a hamming distance at most $1$ from $0^k$ i.e. $\P_k = H_k\cup_{i=1}^{k} e_i\cup 0^k$. 

\subsection{The Distribution \texorpdfstring{$\cald_{k, \epsilon}$}{}}
\label{subsection:distribution}
For $0< \epsilon \leq \frac{1}{k^2}$, consider the following distribution $\cald_{k, \epsilon}$ on the set of satisfying assignments of $\P_k$ where $\alpha := (k-1)\epsilon$.
\begin{align*}
\text{Probabilities} & \hspace{20pt}\text{Assignments}\\
\cald_{k,\epsilon} & \leftarrow \left\{\begin{array}{cccc}
x_1 & x_2 & \cdots\cdots & x_{k}
\end{array}\right.\\
\frac{1}{1-\alpha}\left(\frac{1}{k+1} - \alpha\right)& \leftarrow \left\{\begin{array}{cccc}
0 & 0 & \cdots\cdots & 0
\end{array}\right.\\
\frac{1}{1-\alpha}\left(\frac{1}{k+1} - \epsilon\right) & \leftarrow \left\{ \begin{array}{cccc}
& & h_1 & \\
& & h_2 & \\
& & \vdots &\\
& & h_k &
\end{array} \right.\\
\frac{\epsilon}{1-\alpha} & \leftarrow\left\{ \begin{array}{cccc}
1 & 0 & \cdots\cdots & 0 \\
0 & 1 & \cdots\cdots & 0 \\
& & \vdots &\\
0 & 0 & \cdots\cdots & 1, \end{array} \right.
\end{align*}
where each $h_i$ gets a probability mass $\frac{1}{1-\alpha}(\frac{1}{k+1}-\epsilon)$ and each $e_i$ gets weight $\frac{\epsilon}{1-\alpha}$. The reasoning behind choosing this distribution is as follows: An uniform distribution on $H_k \cup 0^k$ has a property that it is uniform on every single co-ordinate and also pairwise independent. These two properties are very useful proving the soundness guarantee. One more property which we require is that the distribution has to be {\em connected}. In order to achieve this, we add $k$ extra assignment $\{e_1, e_2,\ldots, e_k\}$ and force the distribution to be supported on all $ H_k\cup_{i=1}^{k} e_i\cup 0^k$. Even though by adding extra assignments, we loose the pairwise independent property we make sure that the final distribution is {\em almost} pairwise independent.

We now list down the properties of this distribution which we will use in analyzing the dictatorship test.
\begin{observation}
\label{obs:dist prop}
The distribution $\cald_{k, \epsilon}$ above has the following properties:

\begin{enumerate}
    \item \label{prop:supp} $\cald_{k, \epsilon}$ is supported on $\P_k$.
    \item  \label{prop:uniform marginal} Marginal on every single coordinate is uniform.
    \item  \label{prop:low pairwise correlation} For $i\neq j$, covariance of two variables $x_i, x_j$ sampled form above distribution is: $\mathtt{Cov}[x_i, x_j] = -\frac{\epsilon}{2(1-\alpha)}$.
    \item \label{prop:corr bound} If we view $\calD_{k,\epsilon}$ as a joint distribution on space $\prod_{i=1}^k \calX^{(i)}$ where each $\calX^{(i)} = \{0,1\}$, then for all $i\in [k]$,
  $\rho\left(\calX^{(i)}, \prod_{j\in [k]\setminus \{i\}} \calX^{(j)}; \calD_{k,\epsilon}\right) \leq 1-\frac{\epsilon^2}{2(1-\alpha)^2}$.
\end{enumerate}
\end{observation}

\begin{proof}
We prove each of the observations about the distribution. The first property is straight-forward. To prove ($2$), we compute $\E[x_i]$ as follows.
\begin{align*}
\E[x_i] &= (k+1) \cdot  \frac{1}{1-\alpha}\left(\frac{1}{k+1} - \epsilon\right) \cdot \frac{1}{2} + \frac{\epsilon}{1-\alpha}\\
&= \frac{1-\epsilon(k+1)+ 2\epsilon}{2(1-\alpha)}\\
&= \frac{1}{2}
\end{align*}

Consider the quantity $ \displaystyle \mathop{\E}_{\calD_{k,\epsilon}}[x_i x_j]$. If $x$ is sampled from $0$'s or $e_i$'s, the value is $0$. Moreover, we know that if it is sampled uniformly from $H_k\cup0^k$, it is $1/4$ because of pairwise independence and the above fact.  Therefore, we can write 
\[ \displaystyle \mathop{\E}_{\calD_{k,\epsilon}}[x_i x_j]= (k+1)\frac{1}{1-\alpha}\left(\frac{1}{k+1} - \epsilon\right) \frac{1}{4} \] 
We know that $\displaystyle \mathop{\E}_{\calD_{k,\epsilon}}[x_i] =  \displaystyle \mathop{\E}_{\calD_{k,\epsilon}}[x_j] = 1/2$. Therefore, 
\begin{align*}
\mathtt{Cov}[x_i, x_j] & =  \mathop{\E}_{\calD_{k,\epsilon}}[x_ix_j] -   \mathop{\E}_{\calD_{k,\epsilon}}[x_i]  \mathop{\E}_{\calD_{k,\epsilon}}[x_j]\\
& = \frac{1}{4(1-\alpha)} - \frac{\epsilon(k+1)}{4(1-\alpha)} - \frac{1}{4}\\
& =  \frac{-\epsilon}{2(1-\alpha)}
\end{align*}

To prove the last item, we first show that the bi-partite graph $G\left(\calX^{(i)}, \prod_{j\in [k]\setminus \{i\}} \calX^{(j)}, E\right)$ where $(a,b) \in \calX^{(i)} \times \prod_{j\in [k]\setminus \{i\}} \calX^{(j)}$ is an edge iff $\Pr(a,b)>0$, is connected. To see that the graph is connected, note that for both $0$ and $1$ on the left hand side, $0^{k-1}$ is a neighbor on the right hand side as the distribution's support includes $e_i$ for all $i$, and $0^k$. From the distribution, we see that the smallest atom is at least $\frac{\epsilon}{1-\alpha}$, since $\epsilon \leq 1/k^2$. We now use \lref[Lemma]{lemma:corr_bound} to get the required result.  
\end{proof}

\subsection{Dictatorship Test}
We will switch the notations from $\{0,1\}$ to $\{+1,-1\}$ where we identify $+1$ as $0$ and $-1$ as $1$. 
Let $f :\{-1,+1\}^n \rightarrow \{-1,+1\}$ be a given boolean function. We also assume that $f$ is folded i.e. for every $\v{x}\in \{-1,+1\}^n$, $f(\v{x}) = -f(-\v{x})$.
We think of $\P_k$ as a function $\P_k: \{-1,+1\}^k \rightarrow\{0,1\}$ such that $P_k(z)=1$ iff $z\in \P_k$. Consider the following dictatorship test:

\begin{tcolorbox}[colback=white]
{\bf Test $\calT_{k, \delta}$}
\begin{enumerate}
    \item Sample $\v{x}_1, \v{x}_2,\cdots, \v{x}_k \in \{-1,+1\}^n$ as follows:
  
    \begin{enumerate}
        \item For each $i\in [n]$, independently sample $((\v{x}_1)_i, (\v{x}_2)_i,\cdots, (\v{x}_k)_i)$ according to the distribution $\cald_{k,\delta}$.
        
    \end{enumerate}
    \item Check if $(f(\v{x}_1), f(\v{x}_2),\cdots, f(\v{x}_k)) \in \P_k$.
\end{enumerate}
\end{tcolorbox}

The final test distribution is basically the above test where the parameter $\delta$ is chosen from an appropriate distribution. For a given $\frac{1}{k^2}\geq\epsilon>0$, let  $\err = \frac{\epsilon/5}{2^k}$ and define the following quantities : $\epsilon_0 = \epsilon$ and for $j\geq 0$,  $\epsilon_{j+1} = \err \cdot 2^{- \left( \frac{k^{10}}{\err^3\epsilon_j}\right)^k}$. 

\newcommand{\numtest}{r}
\begin{tcolorbox}[colback=white]
{\bf Test $\calT'_{k, \epsilon}$}

\begin{center}

\begin{enumerate}
	\item Set $\numtest = \left(\frac{k}{\err}\right)^2$
    \item Select $j$ from $\{1,2,\ldots, \numtest\}$ uniformly at random.
	\item Set $\delta =\epsilon_j$
    \item Run test $\calT_{k, \delta}$.
\end{enumerate}
\end{center}
\end{tcolorbox}

We would like to make a remark that this particular setting of $\epsilon_{j+1}$ is not very important. For our analysis, we need a sequence of $\epsilon_j$'s such that each subsequent $\epsilon_j$ is sufficiently small compared to $\epsilon_{j-1}$.

\section{Analysis of the Dictatorship Test}
\label{section:analysis}

\paragraph{Notation:}
We can view $f : \{-1,+1\}^n \rightarrow \{-1,+1\}$ as a function over $n$-fold product set $\X_1\times\X_2\times\cdots\times \X_n$ where each $\X_i = \{-1,+1\}^{\{i\}}$. In the test distribution $\calT_{k, \delta}$, we can think of $\v{x}_i$ sampled from the product distribution on $\X_1^{(i)}\times\X_2^{(i)}\times\cdots\times \X_n^{(i)}$. With these notations in hand, the overall distribution on $(\v{x}_1, \v{x}_2,\cdots, \v{x}_k)$, from the test $\calT_{k, \delta}$,  is a $n$-fold product distribution from the space
$$ \prod_{j=1}^n \left( \prod_{i=1}^k \X_j^{(i)}\right).$$
where we think of $\prod_{i=1}^k \X_j^{(i)}$ as correlated space. We define the parameters for the sake of notational convenience:

\begin{enumerate}
\item $\beta_j := \frac{\epsilon_j}{1-(k-1)\epsilon_j}$ be the minimum probability of an atom in the distribution  $\calD_{k,\epsilon_j}$.
\item $s_{j+1} := \log(\frac{k}{\err})\frac{1}{\epsilon_j^2}$ and $S_j = s_{j+1}$ for $0\leq j\leq \numtest$.
\item $\alpha_j := (k-1)\epsilon_j$ for $j\in [\numtest]$,
\end{enumerate}

\subsection{Completeness}
Completeness is trivial, if $f$ is say $ith$ dictator then the test will be checking the following condition
$$ ((\v{x}_1)_i, (\v{x}_2)_i,\cdots, (\v{x}_k)_i) \in \P_k$$
Using \lref[Observation]{obs:dist prop}(\ref{prop:supp}), the distribution is supported on only strings in $\P_k$. Therefore, the test accepts with probability $1$.

\subsection{Soundness}

\begin{lemma}\label{lemma:soundness}
For every $\frac{1}{k^2}\geq \epsilon>0$ there exists $0< \tau< 1, d \in \mathbf{N}^+$ such that the following holds:
Suppose $f$ is such that for all $i\in [n]$,
$\inf_i^{\leq d}(f) \leq \tau$, then the test $\calT'_{k,\epsilon}$  accepts with probability at most $\frac{2k+1}{2^k} + \epsilon$.
(Note: One can take $\tau$ such that $\tau^{\Omega_k(\err/10s_\numtest\log(1/\beta_\numtest))} \leq \err$ and $d = \frac{\log(1/\tau)}{ \log (1/\beta_\numtest)}$.)
\end{lemma}
\begin{proof}
The acceptance probability of the test is given by the following expression:
\begin{align*}
    \Pr[\mbox{Test accepts } f] &= \E_{\calT'_{k, \epsilon}} [\P_k(f(\v{x}_1), f(\v{x}_2),\cdots, f(\v{x}_k))]
\end{align*}
 After expanding $P_k$ in terms of its Fourier expansion, we get
 \begin{align*}   
  \Pr[\mbox{Test accepts } f]  & = \frac{2k+1}{2^k} + \E_{\calT'_{k, \epsilon}}\left[\sum_{S\subseteq [k], S\neq \emptyset}\hat{\P_k}(S) \prod_{i\in S}f(\v{x}_i)\right]\\
        & = \frac{2k+1}{2^k} + \sum_{S\subseteq [k], S\neq \emptyset}\hat{\P_k}(S) \E_{\calT'_{k, \epsilon}}\left[\prod_{i\in S}f(\v{x}_i)\right]\\
        & \leq \frac{2k+1}{2^k} + \sum_{S\subseteq [k], S\neq \emptyset} \left|\E_{\calT'_{k, \epsilon}}\left[\prod_{i\in S}f(\v{x}_i)\right]\right| \tag{$|\hat{\P_k}(S)|\leq 1$}\\
             & = \frac{2k+1}{2^k} + \sum_{S\subseteq [k], |S|\geq 2} \left|\E_{\calT'_{k, \epsilon}}\left[\prod_{i\in S}f(\v{x}_i)\right]\right|.
\end{align*}
In the last equality, we used the fact that each $\v{x}_i$ is distributed uniformly in $\{-1,+1\}^n$ and hence when $S=\{i\}$, $\E[f(\v{x}_i)] = \hat{f}(\emptyset) =0$. Thus, to prove the lemma it is enough to show that for all $S\subseteq [k]$ such that $|S|\geq 2$,  $\E\left[\prod_{i\in S}f(\v{x}_i)\right]\leq \frac{\epsilon}{2^k}$. This follows from \lref[Lemma]{lemma:main}.
\end{proof}

\begin{lemma}
\label{lemma:main}
For any $S\subseteq[k]$ such that $|S|\geq 2$,

$$ \left|\E_{j\in[\numtest]} \left[ \E_{\cald_{k,\epsilon_j}^{\otimes n}}\left[ \prod_{i\in S}f(\v{x}_i)\right]\right]\right| \leq \frac{\epsilon}{2^k}$$
\end{lemma}
The proof of this follows from the following Lemmas ~\ref{lemma:lowdeg}
, \ref{lemma:invariance}, \ref{lemma:indep_gaussian}.

\begin{lemma}
\label{lemma:lowdeg}
For any $j\in [\numtest]$ and for any $S\subseteq[k]$, $|S|\geq 2$ such that $S=\{\ell_1, \ell_2,\ldots,\ell_t\}$,

$$ \left|  \E_{\cald_{k,\epsilon_j}^{\otimes n}}\left[ \prod_{\ell_i\in S}f(\v{x}_{\ell_i})\right] -  \E_{\cald_{k,\epsilon_j}^{\otimes n}}\left[ \prod_{\ell_i\in S}(T_{1-\gamma_j}f)^{\leq {d_{j,i}}}(\v{x}_{\ell_i})\right] \right| \leq 2\cdot \err + k\sqrt{\sum_{s_j \leq |T|\leq S_j} \hat{f}(T)^2}.$$
where $\gamma_j = \frac{\err}{ks_j}$ and $d_{j,i}$ is a sequence given by  $d_{j,1} =  \frac{2k^2\cdot s_j}{\err}\log\left(\frac{k}{\err}\right)$ and $d_{j,i}= (d_{j,1})^i$ for $1<i\leq t$. 
\end{lemma}

\begin{lemma}
\label{lemma:invariance}
Let $j\in [\numtest]$ and $\nu_j$ be a distribution on jointly distributed standard Gaussian variables with same covariance matrix as that of $\cald_{k,\epsilon_j}$. Then for any $S\subseteq[k]$, $|S|\geq 2$ such that $S=\{\ell_1, \ell_2,\ldots,\ell_t\}$,

$$ \left|  \E_{\cald_{k,\epsilon_j}^{\otimes n}}\left[ \prod_{\ell_i\in S}(T_{1-\gamma_j}f)^{\leq d_{j,i}}(\v{x}_{\ell_i})\right] - \E_{(\v{g}_1, \v{g}_2,\ldots, \v{g}_k)\sim \nu_j^{\otimes n}}\left[ \prod_{\ell_i\in S}(T_{1-\gamma_j}f)^{\leq d_{j,i}}(\v{g}_i)\right]   \right| \leq \err_2$$
where $d_{j,i}$ from \lref[Lemma]{lemma:lowdeg} and $\err_2 = \tau^{\Omega_k(\gamma_j/\log(1/\beta_j))}$ 
(Note: $\Omega(.)$ hides a constant depending on $k$).
\end{lemma}

\begin{lemma}
\label{lemma:indep_gaussian}
Let $k\geq 2$ and $S\subseteq [k]$ such that $|S|\geq 2$ and let $f: \R^n \rightarrow \R$ be a multilinear polynomial of degree $D\geq 1$ such that $\|f\|_2\leq 1$. If $\mathcal{G}$ be a joint distribution on $k$ standard gaussian random variable with a covariance matrix $(1+\delta){\bf I} - \delta {\bf J}$ and $\mathcal{H}$ be a distribution on $k$ independent standard gaussian then it holds that
$$ \left| \E_{\mathcal{G}^{\otimes n}}\left[ \prod_{i\in S}f(\v{g}_i)\right] - \E_{\mathcal{H}^{\otimes n}}\left[ \prod_{i\in S}f(\v{h}_i)\right]   \right| \leq  \delta  \cdot (2k)^{2kD}$$
\end{lemma}
Proofs of  Lemma ~\ref{lemma:lowdeg}
, \ref{lemma:invariance}, \ref{lemma:indep_gaussian} appear in \lref[Section]{section:proofs}. We now prove \lref[Lemma]{lemma:main} using the above three claims.

\vspace{10pt}
\noindent {\bf Proof of \lref[Lemma]{lemma:main}:}
Let $S=\{\ell_1, \ell_2, \ldots, \ell_t\}$. We are interested in getting an upper bound for the following expectation:
$$ \left|\E_{j\in[\numtest]} \left[ \E_{\cald_{k,\epsilon_j}^{\otimes n}}\left[ \prod_{\ell_i\in S}f(\v{x}_{\ell_i})\right]\right]\right| \leq \E_{j\in[\numtest]} \left[\left| \E_{\cald_{k,\epsilon_j}^{\otimes n}}\left[ \prod_{\ell_i\in S}f(\v{x}_{\ell_i})\right]\right|\right]. $$
Let us look at the inner expectation first. Let $\gamma_j = \frac{\err}{ks_j}$ and the sequence $d_{j,i}$ be from \lref[Lemma]{lemma:lowdeg}. We can upper bound the inner expectation as follows:

\begin{align}
 \left| \E_{\cald_{k,\epsilon_j}^{\otimes n}}\left[ \prod_{{\ell_i}\in S}f(\v{x}_{\ell_i})\right] \right| &\leq  \left|\E_{\cald_{k,\epsilon_j}^{\otimes n}}\left[ \prod_{\ell_i\in S}(T_{1-\gamma_j}f)^{\leq {d_{j,i}}}(\v{x}_{\ell_i})\right]\right| + 2\cdot \err + k\sqrt{\sum_{s_j \leq |T|\leq S_j} \hat{f}(T)^2} \tag*{(by \lref[Lemma]{lemma:lowdeg})}\nonumber\\
(\mbox{by \lref[Lemma]{lemma:invariance}}) &\leq \left|\E_{(\v{g}_1, \v{g}_2,\ldots, \v{g}_k)\sim \nu_j^{\otimes n}}\left[ \prod_{\ell_i\in S}(T_{1-\gamma_j}f)^{\leq d_{j,i}}(\v{g}_i)\right] \right|+ \err_2  + 2\cdot \err + k\sqrt{\sum_{s_j \leq |T|\leq S_j} \hat{f}(T)^2},  \label{eq:gau} 
\end{align}
where $\err_2 = \tau^{\Omega_k(\gamma_j/\log(1/\beta_j))}$ and $\nu_j$ has the same covariance matrix as $\calD_{k, \epsilon_j}$. If we let $\delta_j = \frac{2\epsilon_j}{1-\alpha_j}$ then using \lref[Observation]{obs:dist prop}(\ref{prop:low pairwise correlation}), the covariance matrix is precisely $(1+\delta_j){\bf I} - \delta_j {\bf J}$ (note that we switched from $0/1$ to $-1/+1$ which changes the covaraince by a factor of $4$). Each of the functions $(T_{1-\gamma_j}f)^{\leq d_{j,i}}$ has $\ell_2$ norm upper bounded by $1$ and degree at most $d_{j,t}$. We can now apply \lref[Lemma]{lemma:indep_gaussian} to conclude that
\begin{align}
\left|\E_{(\v{g}_1, \v{g}_2,\ldots, \v{g}_k)\sim \nu_j^{\otimes n}}\left[ \prod_{\ell_i\in S}(T_{1-\gamma_j}f)^{\leq d_{j,i}}(\v{g}_i)\right]\right| \leq \left|\E_{(\v{h}_1, \v{h}_2,\ldots, \v{h}_k)}\left[ \prod_{\ell_i\in S}(T_{1-\gamma_j}f)^{\leq d_{j,i}}(\v{h}_i)\right] \right|+ \delta_j  \cdot (2k)^{2k d_{j,t}},\label{eq:gau ind}
\end{align}
where $\v{h}_i$'s are independent and each $\v{h}_i$ is distributed according to $\mathcal{N}(0,1)^n$. Thus,
\begin{align}
\E_{(\v{h}_1, \v{h}_2,\ldots, \v{h}_k)}\left[ \prod_{\ell_i\in S}(T_{1-\gamma_j}f)^{\leq d_{j,i}}(\v{h}_i)\right] &= \prod_{\ell_i\in S} \E_{\v{h}_i}\left[(T_{1-\gamma_j}f)^{\leq d_{j,i}}(\v{h}_i)\right]\nonumber\\
& = \left(\widehat{(T_{1-\gamma_j}f)^{\leq d_{j,i}}}(\emptyset)\right)^t = (\hat{f}(\emptyset))^t = 0,\label{eq:ind zero}
\end{align}
where we used the fact that $f$ is a folded function in the last step. Combining (\ref{eq:gau}), (\ref{eq:gau ind}) and (\ref{eq:ind zero}), we get 
\begin{align}
\left|\E_{\cald_{k,\epsilon_j}^{\otimes n}}\left[ \prod_{{\ell_i}\in S}f(\v{x}_{\ell_i})\right] \right|&\leq  \left(\delta_j  \cdot (2k)^{2k d_{j,t}}\right) + \left(\tau^{\Omega_k(\gamma_j/\log(1/\beta_j))}\right) + 2\cdot \err + k\sqrt{\sum_{s_j \leq |T|\leq S_j} \hat{f}(T)^2}\label{eq:with err terms}
\end{align}
We now upper bound the first term. For this, we use a very generous  upper bounds $d_{j,1} \leq \frac{k^5}{\err^3}\frac{1}{\epsilon_{j-1}^2}$ and $\delta_j \leq 4\epsilon_j$.
\begin{align*}
\delta_j  \cdot (2k)^{2k d_{j,t}}  &\leq   \left({ 4\epsilon_j} \cdot (2k)^{2{\mathbf d_{j,k}} k}\right)\\
 &\leq \epsilon_j \cdot 2^{\left( \frac{k^{10}}{\err^3\epsilon_{j-1}}\right)^k}\\
 &\leq \err.\tag*{{\Bigg (} using $\epsilon_{j} = \err \cdot 2^{- \left( \frac{k^{10}}{\err^3\epsilon_{j-1}}\right)^k}${\Bigg )}}
 \end{align*}
The second term in (\ref{eq:with err terms}) can also be upper bounded by $\err$ by choosing small enough $\tau$.
$$\max_{j}\{\left(\tau^{\Omega_k(\gamma_j/\log(1/\beta_j))}\right) \} \leq \left(\tau^{\Omega_k(\gamma_r/\log(1/\beta_r))}\right) \leq \err.$$
 Finally, taking the outer expectation of (\ref{eq:with err terms}), we get 
 \begin{align*}
  \E_{j\in[\numtest]} \left[ \left| \E_{\cald_{k,\epsilon_j}^{\otimes n}}\left[ \prod_{\ell_i\in S}f(\v{x}_{\ell_i})\right]\right|\right]\leq 4\cdot \err + k\E_{j\in \numtest}\left[\sqrt{\sum_{s_j \leq |T|\leq S_j} \hat{f}(T)^2} \right].
 \end{align*}
  Using Cauchy-Schwartz inequality,
$$\E_{j\in[\numtest]} \left[ \sqrt{\sum_{s_j < |T|< S_j} \hat{f}(T)^2}\right]  \leq \sqrt{\E_{j\in[\numtest]} \left[ \sum_{s_j < |T|< S_j} \hat{f}(T)^2\right] } \leq \frac{1}{\sqrt{\numtest}},$$
where the last inequality uses the fact that the intervals $(s_j , S_j)$ are disjoint for $j\in [\numtest]$ and $\|f\|_2^2 = \sum_T \hat{f}(T)^2 \leq 1$. The final bound we get is 
\begin{align*}
\left|  \E_{j\in[\numtest]} \left[ \E_{\cald_{k,\epsilon_j}^{\otimes n}}\left[ \prod_{\ell_i\in S}f(\v{x}_{\ell_i})\right]\right]\right| \leq  \E_{j\in[\numtest]} \left[ \left|\E_{\cald_{k,\epsilon_j}^{\otimes n}}\left[ \prod_{\ell_i\in S}f(\v{x}_{\ell_i})\right]\right|\right]\leq 4\cdot \err + \frac{k}{\sqrt{\numtest}} \leq 5.\err \leq \frac{\epsilon}{2^k},
 \end{align*}
as required. 
\qed

\section{Proofs of  Lemma ~\ref{lemma:lowdeg}
, \ref{lemma:invariance} \& \ref{lemma:indep_gaussian}}
\label{section:proofs}
In this section, we provide proofs of three crucial lemmas which we used in proving the soundness analysis of our dictatorship test.

\subsection{Moving to a low degree function}
The following lemma, at a very high level, says that if change $f$ to its low degree {\em noisy version} then the loss we incur in the expected quantity is small. 
\begin{lemma}[Restatement of {\lref[Lemma]{lemma:lowdeg}}]
For any $j\in [\numtest]$ and for any $S\subseteq[k]$, $|S|\geq 2$ such that $S=\{\ell_1, \ell_2,\ldots,\ell_t\}$,

$$ \left|  \E_{\cald_{k,\epsilon_j}^{\otimes n}}\left[ \prod_{\ell_i\in S}f(\v{x}_{\ell_i})\right] -  \E_{\cald_{k,\epsilon_j}^{\otimes n}}\left[ \prod_{\ell_i\in S}(T_{1-\gamma_j}f)^{\leq {d_{j,i}}}(\v{x}_{\ell_i})\right] \right| \leq 2\cdot \err + k\sqrt{\sum_{s_j \leq |T|\leq S_j} \hat{f}(T)^2}.$$
where $\gamma_j = \frac{\err}{ks_j}$ and $d_{j,i}$ is a sequence given by  $d_{j,1} =  \frac{2k^2\cdot s_j}{\err}\log\left(\frac{k}{\err}\right)$ and $d_{j,i}= (d_{j,1})^i$ for $1<i\leq t$. 
\end{lemma}

\begin{proof}
The proof is presented in two parts. We first prove an upper bound on
\begin{equation}
\label{eq9}
 \Gamma_1 := \left|  \E_{\cald_{k,\epsilon_j}^{\otimes n}}\left[ \prod_{\ell_i\in S}f(\v{x}_{\ell_i})\right] -  \E_{\cald_{k,\epsilon_j}^{\otimes n}}\left[ \prod_{\ell_i\in S}(T_{1-\gamma_j}f)(\v{x}_{\ell_i})\right] \right| \leq \err + k\sqrt{\sum_{s_j \leq |T|\leq S_j} \hat{f}(T)^2}
 \end{equation}

and then an upper bound on 
\begin{equation}
\label{eq10}
 \Gamma_2:= \left| \E_{\cald_{k,\epsilon_j}^{\otimes n}}\left[ \prod_{\ell_i\in S}(T_{1-\gamma_j}f)(\v{x}_{\ell_i})\right] - \E_{\cald_{k,\epsilon_j}^{\otimes n}}\left[ \prod_{\ell_i\in S}(T_{1-\gamma_j}f)^{\leq {d_{j,i}}}(\v{x}_{\ell_i})\right] \right| \leq \err. 
 \end{equation}
Note that both these upper bounds are enough to prove the lemma.\\

\noindent{\bf Upper Bounding $\Gamma_1$}: The following analysis is very similar to the one in \cite{TY15}, we reproduce it here for the sake of completeness. The first upper bound is obtained by getting the upper bound for the following, for every $a\in [t]$.
\begin{equation}
\label{eq:Gamma_1a}
\Gamma_{1,a}:= \left|  \E_{\cald_{k,\epsilon_j}^{\otimes n}}\left[ \prod_{i\geq a}f(\v{x}_{\ell_i}) \prod_{i<a} (T_{1-\gamma_j}f)(\v{x}_{\ell_i})\right] -  \E_{\cald_{k,\epsilon_j}^{\otimes n}}\left[\prod_{i>a}f(\v{x}_{\ell_i}) \prod_{i\leq a} (T_{1-\gamma_j}f)(\v{x}_{\ell_i}) \right] \right|
\end{equation}
Note that by triangle inequality, $\Gamma_1 \leq \sum_{a\in[t]} \Gamma_{1,a}$.

\begin{align}
(\ref{eq:Gamma_1a})&= \left|\E_{\cald_{k,\epsilon_j}^{\otimes n}}\left[\left(f(\v{x}_{\ell_a}) - T_{1-\gamma_j}f(\v{x}_{\ell_a})\right) \prod_{i>a}f(\v{x}_{\ell_i})\prod_{i<a} (T_{1-\gamma_j}f)(\v{x}_{\ell_i}) \right] \right| \nonumber\\
&= \left|\E_{\cald_{k,\epsilon_j}^{\otimes n}}\left[\left(id - T_{1-\gamma_j}\right)f(\v{x}_{\ell_a}) \prod_{i>a}f(\v{x}_{\ell_i})\prod_{i<a} (T_{1-\gamma_j}f)(\v{x}_{\ell_i}) \right] \right| \nonumber\\
&= \left|\E_{\cald_{k,\epsilon_j}^{\otimes n}}\left[U\left((id - T_{1-\gamma_j}\right)f) (\v{x}_{\{\ell_i:i\in [t]\setminus\{a\}\}}) \prod_{i > a}f(\v{x}_{\ell_i})\prod_{i<a} (T_{1-\gamma_j}f)(\v{x}_{\ell_i}) \right] \right| \label{eq1}
\end{align}
where $U$ is the Markov operator for the correlated probability space which maps functions from the space $\X^{(\ell_a)}$ to the space $\prod_{i\in [t]\setminus \{a\} }\X^{(\ell_i)}$. We can look at the above expression as a product of two  functions, $F =  \prod_{i>a}f \prod_{i<a} (T_{1-\gamma_j}f)$ and $G = U(id - T_{1-\gamma_j})f)$.  From \lref[Observation]{obs:dist prop}(~\ref{prop:corr bound}), the correlation between spaces $\left(\X^{(\ell_a)}, \prod_{i\in [t]\setminus \{a\} }\X^{(\ell_i)}\right)$ is upper bounded by $1-\left(\frac{\epsilon_j}{1-\alpha_j}\right)^2 \leq 1-\epsilon_j^2 =:\rho_j$. Taking the Efron-Stein decomposition with respect to the product distribution, we have the following because of orthogonality of the Efron-Stein decomposition,
\begin{align}
(\ref{eq1}) = \left|\E_{\cald_{k,\epsilon_j}^{\otimes n}}\left[G \times F \right] \right| &= \left|\sum_{T\subseteq [n]}\E_{\cald_{k,\epsilon_j}^{\otimes n}}\left[G_T \times F_T \right] \right| \nonumber\\
\mbox{ (by Cauchy-Schwartz) } &\leq \sqrt{ \sum_{T\subseteq[n]} \|F_T\|_2^2 }\sqrt{ \sum_{T\subseteq[n]} \|G_T\|_2^2 }\label{eq2}
\end{align}
where the norms are with respect to $\cald_{k,\epsilon_j}^{\otimes n}$'s marginal distribution on the product distribution $\prod_{i \in [t]\setminus \{a\} }\X^{(\ell_i)}$.  By orthogonality, the quantity $\sqrt{ \sum_{T\subseteq[n]} \|F_T\|_2^2 }$ is just $\|F\|_2$. As $F$ is product of function whose range is $[-1,+1]$, rane of $F$ is also $[-1,+1]$ and hence $\|F\|_2$ is at most 1. Therefore, 
\begin{align}
\label{eq3}
&(\ref{eq2}) \leq \sqrt{ \sum_{T\subseteq[n]} \|G_T\|_2^2 }
\end{align}

We have $G_T = (UG')_T$, where $G' = (id - T_{1-\gamma_j})f$. In $G'_T$, the Efron-Stein decomposition is with respect to the marginal distribution of $\cald_{k,\epsilon_j}^{\otimes n}$ on $\X^{(\ell_a)}$, which is just uniform (by \lref[Observation]{obs:dist prop}(\ref{prop:uniform marginal})).  Using \lref[Proposition]{prop:mossel_prop211}, we have $G_T$ = $UG'_T = U(id - T_{1-\gamma_j})f_T$. Substituting in (\ref{eq3}), we get
\begin{equation}
\label{eq4}
(\ref{eq3})= \sqrt{ \sum_{T\subseteq[n]} \|U(if - T_{1-\gamma_j})f_T)\|_2^2 }
\end{equation}

We also have that the correlation is upper bounded by $\rho_j$. We can therefore apply \lref[Proposition]{prop:mossel_prop212}, and conclude that for each $T\subseteq[n]$, 
\[\|U(id - T_{1-\gamma_j})f_T\|_2 \leq \rho_j^{|T|}\|(id - T_{1-\gamma_j})f_T\|_2\]
where the norm on the right is with respect to the uniform distribution. Observe that
\[\|(id- T_{1-\gamma_j})f_T\|_2^2 = (1-(1-\gamma_j)^{|T|})^2\hat{f}(T)^2\]
Substituting back into (\ref{eq4}), we get
\begin{equation}
\label{eq5}
(\ref{eq4})\leq \sqrt{ \sum_{T\subseteq[n]} \underbrace{\rho_j^{2|T|} (1-(1-\gamma_j)^{|T|})^2\hat{f}(T)^2}_{\term(\epsilon_j, \gamma_j, T)} }
\end{equation}
We will now break the above summation into three different parts and bound each part separately.
\begin{align*} 
\Theta_1 &:= \sum_{\substack{T\subseteq [n] ,\\ |T|\leq s_j}} \term(\epsilon_j, \gamma_j, T)  & \Theta_2:=\sum_{\substack{T\subseteq [n] ,\\ s_j<|T| <  S_j}} \term(\epsilon_j, \gamma_j, T)\\
\Theta_3 &:= \sum_{\substack{T\subseteq [n] ,\\ |T|\geq S_j}} \term(\epsilon_j, \gamma_j, T) 
\end{align*}
\begin{itemize}
\item {\bf Upper bounding $\Theta_1$:} 
\begin{align*}
\Theta_1 &= \sum_{\substack{T\subseteq [n] ,\\ |T|\leq s_j}} \term(\epsilon_j, \gamma_j, T)  =\sum_{\substack{T\subseteq [n] ,\\ |T|\leq s_j}}  \rho_j^{2|T|}(1-(1-\gamma_j)^{|T|})^2 \hat{f}(T)^2\leq \sum_{\substack{T\subseteq [n] ,\\ |T|\leq s_j}} (1-(1-\gamma_j)^{|T|})^2 \hat{f}(T)^2.
\end{align*}
For every $|T|\leq s_j$ we have $1-(1-\gamma_j)^{|T|} \leq \err_1/k$. Thus,
$$\Theta_1 \leq \left(\frac{\err_1}{k}\right)^2\sum_{\substack{T\subseteq [n] ,\\ |T|\leq s_j}} \hat{f}(T)^2.$$
\item {\bf Upper bounding $\Theta_3$:} 
\begin{align*}
\Theta_3 &= \sum_{\substack{T\subseteq [n] ,\\ |T|\geq S_j}} \term(\epsilon_j, \gamma_j, T)  =\sum_{\substack{T\subseteq [n] ,\\ |T|\geq S_j}}  \rho_j^{2|T|}(1-(1-\gamma_j)^{|T|})^2 \hat{f}(T)^2\leq \sum_{\substack{T\subseteq [n] ,\\ |T| \geq S_j}} \rho_j^{2|T|} \hat{f}(T)^2.
\end{align*}
For every $|T| \geq S_j$ we have $\rho_j^{|T|}\leq (1-\epsilon_j^2)^{|T|} \leq \err_1/k$. Thus,
$$\Theta_3 \leq \left(\frac{\err_1}{k}\right)^2\sum_{\substack{T\subseteq [n] ,\\ |T| \geq S_j}} \hat{f}(T)^2.$$
\end{itemize}
Substituting these upper bounds in (\ref{eq5}), 
\begin{align*}
\Gamma_{1,a}& \leq \sqrt{ \left(\frac{\err_1}{k}\right)^2 \sum_{\substack{T\subseteq [n] ,\\ |T|\leq s_j or |T|\geq S_j}}  \hat{f}(T)^2 + 
\sum_{\substack{T\subseteq [n] ,\\ s_j < |T|< S_j}}  \hat{f}(T)^2 } \\
& \leq  \sqrt{\left(\frac{\err_1}{k}\right)^2 + \sum_{ s_j < |T|< S_j}  \hat{f}(T)^2} \tag*{(since $\sum_{T} \hat{f}(T)^2 \leq 1$)}\\
& \leq  \frac{\err_1}{k} + \sqrt{\sum_{ s_j < |T|< S_j}  \hat{f}(T)^2}. \tag*{(using concavity)}
\end{align*}
The required upper bound on $\Gamma_1$ follows by using $\Gamma_1 \leq \sum_{a\in[t]} \Gamma_{1,a}$ and the above bound.\\

\noindent{\bf Upper Bounding $\Gamma_2$}: We will now show an upper bound on $\Gamma_2$.  The approach is similar to the previous case, we upper bound the following quantity for every $a\in [t]$ 

{\small $$\Gamma_{2,a}:=\left|  \E_{\cald_{k,\epsilon_j}^{\otimes n}}\left[ \prod_{i\geq a}(T_{1-\gamma_j}f)(\v{x}_{\ell_i}) \prod_{i< a} (T_{1-\gamma_j}f^{\leq d_{j,i}})(\v{x}_{\ell_i})\right] -  \E_{\cald_{k,\epsilon_j}^{\otimes n}}\left[\prod_{i > a}(T_{1-\gamma_j}f)(\v{x}_{\ell_i}) \prod_{i\leq a} (T_{1-\gamma_j}f^{\leq d_{j,i}})(\v{x}_{\ell_i}) \right] \right| \nonumber $$}
\begin{align}
&= \left|\E_{\cald_{k,\epsilon_j}^{\otimes n}}\left[\left(T_{1-\gamma_j}f(\v{x}_{\ell_a}) - T_{1-\gamma_j}f^{\leq d_{j,a}}(\v{x}_{\ell_a})\right) \prod_{i>a}T_{1-\gamma_j}f(\v{x}_{\ell_i})\prod_{i<a} (T_{1-\gamma_j}f^{\leq d_{j,i}})(\v{x}_{\ell_i}) \right] \right| \nonumber\\
&= \left|\E_{\cald_{k,\epsilon_j}^{\otimes n}}\left[\left(T_{1-\gamma_j}f^{> d_{j,a}}(\v{x}_{\ell_a}) \right) \prod_{i>a}T_{1-\gamma_j}f(\v{x}_{\ell_i})\prod_{i<a} (T_{1-\gamma_j}f^{\leq d_{j,i}})(\v{x}_{\ell_i}) \right] \right|  \label{eq6}
\end{align}
By using Holder's inequality we can upper bound (\ref{eq6}) as:
\begin{align}
(\ref{eq6}) \leq \|T_{1-\gamma_j}f^{> d_{j,a}}\|_2\prod_{i>a} \|T_{1-\gamma_j}f\|_{2(t-1)} \prod_{i<a} \|T_{1-\gamma_j}f^{\leq d_{j,i}}\|_{2(t-1)},\label{eq7}
\end{align}
where each norm is w.r.t the uniform distribution as marginal of each $\v{x}_{\ell_i}$ is uniform in $\{+1,-1\}^n$. Now, $\|T_{1-\gamma_j}f\|_{2(t-1)} \leq 1$ as the range if $T_{1-\gamma_j}f$ is in $[-1,+1]$. To upper bound $\|T_{1-\gamma_j}f^{\leq d_{j,i}}\|_{2(t-1)}$, we use \lref[Proposition]{prop:hyper} and using the fact that $\{-1,+1\}$ uniform random variable is $(2,q,1/\sqrt{q-1})$ hypercontractive (\lref[Theorem]{thm:hc param}) to get 
$$\|T_{1-\gamma_j}f^{\leq d_{j,i}}\|_{2(t-1)}\leq (2t-3)^{d_{j,i}}\|T_{1-\gamma_j}f^{\leq d_{j,i}}\|_2 \leq (2t)^{d_{j,i}}.$$
Plugging this in (\ref{eq7}), we get 
\begin{align}
(\ref{eq7}) &\leq \|T_{1-\gamma_j}f^{> d_{j,a}}\|_2 \prod_{i < a}(2t)^{d_{j,i}} \leq (1-\gamma_j)^{d_{j,a}}\cdot \prod_{i < a}(2t)^{d_{j,i}} \nonumber\\
&\leq e^{-\gamma_j d_{j,a}}\cdot (2k)^{k\cdot d_{j,a-1}}\nonumber\\
&\leq e^{-\frac{\err}{ks_j}\cdot d_{j,a}}\cdot (2k)^{k\cdot d_{j,a-1}} \label{eq8}
\end{align}
Now, 
\begin{align*}
d_{j,1}\cdot d_{j,a-1} &= d_{j,a}\\
 \frac{2k^2\cdot s_j}{\err}\log\left(\frac{k}{\err}\right)\cdot d_{j,a-1} & =  d_{j,a}\\
\frac{k^2\cdot s_j}{\err}\log\left(\frac{k}{\err}\right) + \frac{k^2\cdot s_j}{\err}\log\left(\frac{k}{\err}\right)\cdot d_{j,a-1} &\leq d_{j,a}\\
\frac{k\cdot s_j}{\err}\log\left(\frac{k}{\err}\right) + \frac{k^2\cdot s_j}{\err} \cdot\log(2k) \cdot d_{j,a-1} &\leq d_{j,a}\\
\frac{k\cdot s_j}{\err}\cdot \left(\log\left(\frac{k}{\err}\right) + k\cdot d_{j,a-1}\log(2k)\right) &= d_{j,a}\\
\frac{k\cdot s_j}{\err}\cdot \log\left(\frac{k}{\err}(2k)^{k\cdot d_{j,a-1}}\right) &= d_{j,a}
\end{align*}
This implies
\begin{align*}
 \log\left(\frac{k}{\err}(2k)^{k\cdot d_{j,a-1}}\right) &= \frac{\err}{ks_j}\cdot d_{j,a}\\
\implies\frac{k}{\err}(2k)^{k\cdot d_{j,a-1}} &= e^{\frac{\err}{ks_j}\cdot d_{j,a}}\\
\implies e^{-\frac{\err}{ks_j}\cdot d_{j,a}}\cdot (2k)^{k\cdot d_{j,a-1}} &= \frac{\err}{k}.
\end{align*}
Thus from (\ref{eq8}), we have $\Gamma_{2,a} \leq \frac{\err}{k}$. To conclude the proof, by triangle inequality we have $\Gamma_2 \leq \sum_{a\in [t]} \Gamma_{2,a} \leq \err$.
\end{proof}

\subsection{Moving to the Gaussian setting}

We are now in the setting of {\em low degree} polynomials because of \lref[Lemma]{lemma:lowdeg}. The following lemma let us switch from our test distribution to a Gaussian distribution with the same first two moments. 

\begin{lemma}[Restatement of {\lref[Lemma]{lemma:invariance}}]
Let $j\in [\numtest]$ and $\nu_j$ be a distribution on jointly distributed standard Gaussian variables with same covariance matrix as that of $\cald_{k,\epsilon_j}$. Then for any $S\subseteq[k]$, $|S|\geq 2$ such that $S=\{\ell_1, \ell_2,\ldots,\ell_t\}$,

$$ \left|  \E_{\cald_{k,\epsilon_j}^{\otimes n}}\left[ \prod_{\ell_i\in S}(T_{1-\gamma_j}f)^{\leq d_{j,i}}(\v{x}_{\ell_i})\right] - \E_{(\v{g}_1, \v{g}_2,\ldots, \v{g}_k)\sim \nu_j^{\otimes n}}\left[ \prod_{\ell_i\in S}(T_{1-\gamma_j}f)^{\leq d_{j,i}}(\v{g}_i)\right]   \right| \leq \err_2$$
where $d_{j,i}$ from \lref[Lemma]{lemma:lowdeg} and $\err_2 = \tau^{\Omega_k(\gamma_j/\log(1/\beta_j))}$ 
(Note: $\Omega(.)$ hides a constant depending on $k$).
\end{lemma}
\begin{proof}
Using the definition of $(d, \tau)$-quasirandom function and \lref[Fact]{fact:noise fd}, if $f$ is $(d, \tau)$- quasirandom then so is $T_{1-\gamma}f$ for any $0\leq \gamma \leq 1$.  Also, $T_{1-\gamma}f$ satisfies 
$$\Var[T_{1-\gamma}f^{>d}] = \sum_{\substack{T\subseteq [n] \\ |T|> d}} (1-\gamma)^{2|T|}\hat{f}(T)^2 \leq (1-\gamma)^{2d} \cdot \sum_{\substack{T\subseteq [n] \\ |T|> d}} \hat{f}(T)^2 \leq  (1-\gamma)^{2d}.$$
The lemma  follows from a direct application of \lref[Theorem]{thm:invariance}.
\end{proof}

\subsection{ Making Gaussian variables independent}

Our final lemma allows us to make the Gaussian variables independent. Here we crucially need the property that the polynomials we are dealing with are low degree polynomials. Before proving \lref[Lemma]{lemma:indep_gaussian}, we need the following lemma which says that low degree functions are robust to small perturbations in the input on average.  
\begin{lemma}
\label{lemma:gclose}
Let $f: \R^n \rightarrow \R$ be a multilinear polynomial of degree $d$ such that $\|f\|_2\leq 1$ suppose $\v x, \v z \sim \calN(0,1)^n$ be $n$-dimensional standard gaussian vectors such that $\E[x_i z_i]\geq 1- \delta$ for all $i\in [n]$. Then 
$$ \E[(f(\v x ) -f(\v z))^2] \leq 2 \delta d .$$
\end{lemma}
\begin{proof}
For $T\subseteq[n]$, we have
$$\E[\chi_T(\v x)\chi_T(\v z)] = \prod_{i\in T} \E[x_iz_i]\geq \prod_{i\in T}(1- \delta) \geq  (1-\delta)^{|T|}$$
We now bound the following expression,
\begin{align*}
\E[ ( f(\v x) - f(\v z) )^2] &= \E[f(\v x)^2 + f(\v z)^2 -2f(\v x)z(\v x)]\\
&= \sum_{T\subseteq[n], |T|\leq d} \hat{f}(T)^2 ( 2-2\E[\chi_T(\v x)\chi_T(\v z)])\\
& \leq 2 \cdot \sum_{T\subseteq[n], |T|\leq d} \hat{f}(T)^2 ( 1-(1-\delta)^{|T|})\\
&\leq 2 \cdot \sum_{T\subseteq[n], |T|\leq d} \hat{f}(T)^2 \delta |T| \\
&\leq 2 \delta d\cdot \sum_{T\subseteq[n], |T|\leq d} \hat{f}(T)^2 \leq 2\delta d, 
\end{align*}
where the last inequality uses $\|f\|_2 \leq 1$.
\end{proof}

We are now ready to prove \lref[Lemma]{lemma:indep_gaussian}.

\begin{lemma}[Restatement of {\lref[Lemma]{lemma:indep_gaussian}}]
Let $k\geq 2$ and $2\leq t\leq k$ and let $f: \R^n \rightarrow \R$ be a multilinear polynomial of degree $D\geq 1$ such that $\|f\|_2\leq 1$. If $\mathcal{G}$ be a joint distribution on $k$ standard gaussian random variable with covariance matrix $(1+\delta) {\bf I} - \delta {\bf J}$ and $\mathcal{H}$ be a distribution on $k$ independent standard gaussian then it holds that
$$ \left| \E_{\mathcal{G}^{\otimes n}}\left[ \prod_{i\in [t]}f(\v{g}_i)\right] - \E_{\mathcal{H}^{\otimes n}}\left[ \prod_{i\in [t]}f(\v{h}_i)\right]   \right| \leq \delta\cdot (2k)^{2Dk}.$$
\end{lemma}
\begin{proof} 
Let $\v{\Sigma}=(1+\delta) {\bf I} - \delta {\bf J}$ be the covariance matrix. Let ${\bf M} = (1-\delta')((1+\beta) {\bf I} - \beta {\bf J})$ be a matrix such that $\bf M^2 = \Sigma$. There are multiple ${\bf M}$ which satisfy $\bf M^2 = \Sigma$. We chose the ${\bf M}$ stated above to make the analysis simpler. From the way we chose ${\bf M}$ and using the condition $\bf M^2 = \Sigma$, it is easy to observe that $\beta$ and $\delta'$ should satisfy the following two conditions:
$$1-\delta' = \frac{1}{\sqrt{1+(k-1)\beta^2}}\quad \text{and} \quad \frac{(k-2)\beta^2-2\beta}{1+(k-1)\beta^2} =-\delta.$$
Since $\mathcal{H}$ is a distribution of $k$ independent standard gaussians, we can generate a sample $x\sim \calG$ by sampling $y \sim \mathcal{H}$ and setting $x = {\bf M}y$. In what follows, we stick to the following notation: $(\v{h}_1, \v{h}_2, \ldots, \v{h}_k)\sim \calH^{\otimes n}$  and $(\v{g}_1, \v{g}_2, \ldots, \v{g}_k)_j = {\bf M} (\v{h}_1, \v{h}_2, \ldots, \v{h}_k)_j$ for each $j\in [n]$.

Because of the way we chose to generate $g_i's$, we have for all $i\in [k]$ and $j\in [n]$, $\E[(\v{g}_i)_j(\v{h}_i)_j] = 1-\delta'\geq 1- k\beta^2$. To get an upper bound on $\beta$, notice that $\beta$ is a root of the quadratic equation $(k + \delta k - \delta - 2) \beta^2 - 2\beta + \delta=0$.  Let $k'= (k + \delta k - \delta - 2)$, if $\beta_1, \beta_2$ are the roots of the equation then they satisfy: $k'\beta_1 + k'\beta_2 = 2$ and $(k'\beta_1)(k'\beta_2) = \delta k'$ and $\beta_1, \beta_2>0$. Thus, we have $\min\{ k'\beta_1, k'\beta_2\} \leq \delta k'$ and hence, we can take $\beta$ such that $\beta\leq \delta$.

 We wish to upper bound the following expression:

$$\Gamma:= \left| \E_{\mathcal{H}^{\otimes n}}\left[ \prod_{i\in [t]}f(\v{g}_i) - \prod_{i\in [t]}f(\v{h}_i)\right]   \right|.$$
Define the following quantity
$$\Gamma_i := \left| \E_{\mathcal{H}^{\otimes n}}\left[ \prod_{j=1}^{i-1} f(\v{h}_j)\prod_{j=i}^{t} f(\v{g}_j)  -  \prod_{j=1}^{i} f(\v{h}_j) \prod_{j=i+1}^{t} f(\v{g}_j) \right]   \right|.$$
By triangle inequality, we have $\Gamma \leq \sum_{i\in [t]}\Gamma_i$. We now proceed with upper bounding $\Gamma_i$ for a given $i\in[t]$.

\begin{align*}
\Gamma_i &= \left| \E_{\mathcal{H}^{\otimes n}}\left[ \prod_{j=1}^{i-1} f(\v{h}_j)\prod_{j=i}^{t} f(\v{g}_j) - \prod_{j=1}^{i} f(\v{h}_j) \prod_{j=i+1}^{t} f(\v{g}_j) \right]   \right|\\
&=\left| \E_{\mathcal{H}^{\otimes n}}\left[ (f(\v{g}_i) - f(\v{h}_i) )\cdot \prod_{j=1}^{i-1} f(\v{h}_j)\prod_{j=i+1}^{t} f(\v{g}_j) \right]  \right|\\
& \leq \sqrt{\E_{\mathcal{H}^{\otimes n}}[ (f(\v{g}_i) - f(\v{h}_i))^2]} \cdot  \prod_{j=1}^{i-1} \E_{\mathcal{H}^{\otimes n}}[f(\v{h}_j)^{2(t-1)}]^{\frac{1}{2(t-1)}} \prod_{j=i+1}^{t} \E_{\mathcal{H}^{\otimes n}}[ f(\v{g}_j) ^{2(t-1)}]^{\frac{1}{2(t-1)}},
\end{align*}
where the last step uses Holder's Inequality. Now, the marginal distribution on each $h_j$ and $g_j$ is identical which is $\mathcal{N}(0,1)^n$, we have
\begin{align*}
\Gamma_i  &\leq \sqrt{\E_{\mathcal{H}^{\otimes n}}[ (f(\v{g}_i) - f(\v{h}_i))^2]} \cdot  \prod_{j=1}^{i-1} \|f\|_{2(t-1)} \prod_{j=i+1}^{t} \|f\|_{2(t-1)}\\
& \leq \sqrt{\E_{\mathcal{H}^{\otimes n}}[ (f(\v{g}_i) - f(\v{h}_i))^2]} \cdot ( \|f\|_{2(t-1)})^{t-1}
\end{align*}
Since a standard one dimensional Gaussian is $(2, q, 1/\sqrt{q-1})$-hypercontractive (\lref[Theorem]{thm:hc param}), from \lref[Proposition]{prop:hyper} , $\|f\|_{2(t-1)} \leq (\sqrt{2t-3})^D\|f\|_2 \leq (\sqrt{2t-3})^D < (2t)^{D/2} $.
Thus,
$$\Gamma_i \leq  (2t)^{D(t-1)/2}  \cdot \sqrt{\E_{\mathcal{H}^{\otimes n}}[ (f(\v{g}_i) - f(\v{h}_i))^2]}$$
Now, each $\v{g}_i, \v{h}_i$ are such that such that $\E[(\v{g}_i)_j\cdot(\v{h}_i)_j] = 1-\delta'\geq 1-k\delta^2$ for every $j\in [n]$. We can apply \lref[Lemma]{lemma:gclose} to get $\E_{\mathcal{H}^{\otimes n}}[ (f(\v{g}_i) - f(\v{h}_i))^2] \leq 2k\delta^2 D$. Hence, we can safely upper bound $\Gamma_i$ as 
$$\Gamma_i \leq (2t)^{D(t-1)/2}  \cdot 2k\delta D.$$ 
Therefore, $\Gamma \leq \sum_i \Gamma_i \leq t\cdot (2t)^{D(t-1)/2}  \cdot 2k\delta D$ which is at most $2k^2\delta D \cdot (2k)^{Dk/2} \leq  \delta \cdot (2k)^{2Dk}$ as required.
\end{proof}

	{\small
		\bibliographystyle{alpha}
		\bibliography{covering-bib}
	}
\end{document}